\documentclass[sigconf]{acmart}

\AtBeginDocument{%
  }

\setcopyright{acmlicensed}
\copyrightyear{2025}
\acmYear{2025}
\setcopyright{acmlicensed}
\acmConference[KDD '25] {Proceedings of the 31st ACM SIGKDD Conference on Knowledge Discovery and Data Mining V.2}{August 3--7, 2025}{Toronto, ON, Canada.}
\acmBooktitle{Proceedings of the 31st ACM SIGKDD Conference on Knowledge Discovery and Data Mining V.2 (KDD '25), August 3--7, 2025, Toronto, ON, Canada}
\acmISBN{979-8-4007-1454-2/25/08}
\acmDOI{10.1145/3711896.3737043}

\usepackage{algorithm}
\usepackage{algorithmic}

\usepackage{mathrsfs}
\usepackage{amssymb}
\usepackage{bm}
\usepackage{amsmath}
\usepackage{dsfont}
\usepackage{booktabs}
\usepackage{epstopdf}
\usepackage{multirow}
\usepackage{comment}
\usepackage{endnotes}
\usepackage{hyperref}

\usepackage{bibentry}
\usepackage[normalem]{ulem}
\usepackage{graphicx}
\usepackage{enumerate}
\usepackage{enumitem}
\usepackage{subcaption}
\newtheorem{assumption}{Assumption}
\newtheorem{theorem}{Theorem}

\newtheorem{lemma}[theorem]{Lemma}
\newtheorem{remark}{Remark}

\usepackage{titlesec}
\usepackage{color, xspace}
\titleformat{\section}{\normalfont\Large\bfseries\MakeUppercase}{\thesection}{1em}{}

\begin{document}

\title{Measure Domain's Gap: A Similar Domain Selection Principle for Multi-Domain Recommendation}

\author{Yi Wen}
\authornote{This work was done when Yi Wen was a visiting student at Beihang University.}
\email{yiwen23-c@my.cityu.edu.hk}
\affiliation{%
  \institution{City University of Hong Kong}
  \city{Hong Kong SAR}
  \country{China}
}

\author{Yue Liu}
\email{yliu@u.nus.edu}
\affiliation{%
  \institution{National University of Singapore}
  \country{Singapore}}

\author{Derong Xu}
\email{derongxu2-c@my.cityu.edu.hk}
\affiliation{%
  \institution{City University of Hong Kong}
  \city{Hong Kong SAR}
  \country{China}
}

\author{Huishi Luo}
\email{hsluo2000@buaa.edu.cn}
\affiliation{%
  \institution{Institute of Artificial Intelligence, Beihang University}
  \city{Beijing}
  \country{China}
}

\author{Pengyue Jia}
\email{jia.pengyue@my.cityu.edu.hk}
\affiliation{%
  \institution{City University of Hong Kong}
  \city{Hong Kong SAR}
  \country{China}
}

\author{Yiqing Wu}
\email{iwu_yiqing@163.com}
\affiliation{%
  \institution{Institute of Computing Technology,
Chinese Academic of Science}
  \city{Beijing}
  \country{China}
}

\author{Siwei Wang}
\email{wangsiwei13@nudt.edu.cn}
\affiliation{%
  \institution{Intelligent Game and Decision Lab}
  \city{Beijing}
  \country{China}
}

\author{Ke Liang}
\email{liangke200694@126.com}
\affiliation{%
  \institution{National University of Defense Technology}
  \city{Changsha}
  \country{China}
}

\author{Maolin Wang}
\authornotemark[2]
\email{morin.wang@my.cityu.edu.hk}
\affiliation{%
  \institution{City University of Hong Kong}
  \city{Hong Kong SAR}
  \country{China}
}

\author{Yiqi Wang}
\email{yiq@nudt.edu.cn}
\affiliation{%
  \institution{National University of Defense Technology}
  \city{Changsha}
  \country{China}
}

\author{Fuzhen Zhuang}
\authornote{Corresponding authors}
\email{zhuangfuzhen@buaa.edu.cn}
\affiliation{%
  \institution{Beihang University, State Key Laboratory of Complex \& Critical Software Environment}
  \city{Beijing}
  \country{China}
}

\author{Xiangyu Zhao}
\authornotemark[2]
\email{xianzhao@cityu.edu.hk}
\affiliation{%
  \institution{City University of Hong Kong}
  \city{Hong Kong SAR}
  \country{China}
}
\renewcommand{\shortauthors}{Wen et al.}

\begin{abstract}
Multi-Domain Recommendation (MDR) achieves the desirable recommendation performance by effectively utilizing the transfer information  across different domains. Despite the great success, most existing MDR methods adopt a single structure to transfer complex domain-shared knowledge. 
However, the beneficial transferring information should vary across different domains. When there is knowledge conflict between domains or a domain is of poor quality, unselectively leveraging information from all domains will lead to a serious Negative Transfer Problem (NTP). 
Therefore, how to effectively model the complex transfer relationships between domains to avoid NTP is still a direction worth exploring. To address these issues, we propose a simple and dynamic \textbf{S}imilar \textbf{D}omain \textbf{S}election \textbf{P}rinciple (SDSP) for multi-domain recommendation in this paper\footnote{\url{https://github.com/Applied-Machine-Learning-Lab/SDSP}}. SDSP presents the initial exploration of selecting suitable domain knowledge for each domain to alleviate NTP. Specifically, we propose a novel prototype-based domain distance measure to effectively model the complexity relationship between domains. Thereafter, 
 the proposed SDSP can dynamically find similar domains for each domain based on the supervised signals of the domain metrics and the unsupervised distance measure from the learned domain prototype. We emphasize that SDSP is a lightweight method that can be incorporated with existing MDR methods for better performance while not introducing excessive time overheads. To the best of our knowledge, it is the first solution that can explicitly measure domain-level gaps and dynamically select appropriate domains in the MDR field. Extensive experiments on three datasets demonstrate the effectiveness of our proposed method. 
\end{abstract}
\vspace{-10 pt}

\ccsdesc[500]{Information systems~Recommender systems}

\keywords{Multi-domain Recommendation; Negative Transfer}

    \maketitle
\section{Introduction}
Recommendation systems \cite{kang2018self-seq-rec,zhou2018din-interest-rec,sun2019bert4rec-seq-rec, zhao2018deep-rl-rec, zhao2018recommendations-rec} have achieved remarkable success in a wide range of application fields \cite{ma2018entire} and have become the cornerstone of personalized user experience in industries such as e-commerce \cite{schafer2001commerce-rec}, streaming media services \cite{covington2016deep}, and online education \cite{urdaneta2021recommendation-education-rec}. However, traditional recommendation methods \cite{wang2017deep&cross-rec, guo2017deepfm-rec,cheng2016wide&deep-rec,guan2022cross} can only handle data from a single domain, whereas in the real world, we often need to handle multi-domain data due to the complexity of business needs \cite{xia2021knowledge}. Therefore, it is necessary to model user behavior and item attributes for each domain to capture their different characteristics effectively. However, owing to the considerable disparity in data distribution across different domains, simply concatenating multi-domain data can introduce the domain bias problem \cite{liu2023deep,ni2019justifying,yang2024generate}, and modeling multi-domain data separately imposes significant overheads and maintenance costs, which are unacceptable in the industrial field \cite{he2020lightgcn-graph-rec}. Therefore, how to design a model that can concurrently cope with multi-domain data has become a crucial factor in solving the multi-domain problem.

To address these challenges, Multi-Domain Recommendation (MDR) methods \cite{li2023hamur-multi-domain-rec,park2024pacer-cross-domain-rec,li2023adl} have emerged as a promising approach. By migrating knowledge from multiple domains, the MDR approaches significantly improve the recommendation performance and alleviate the cold-start and data sparsity problems that are difficult to deal with in traditional recommendations. Dynamic Weights (DW) \cite{yan2022apg-rec, bian2020can-rec} and Share-Specific (SS) based methods \cite{tang2020ple-multi-task-rec,park2024pacer-cross-domain-rec}, the two most representative classes of methods in the MDR field, have received increasing attention recently. DW-based approaches input scenario-sensitive features into the network to generate parameters for the main network. SS-based approaches often divide the network or parameters into domain-specific and domain-shared parts to capture shared and specific information among domains.

Despite the significant success achieved, DW and SS-based methods still suffer some unsatisfactory factors. Firstly, DW-based methods rely on a strong prior knowledge of the domains to manually select scenario-sensitive features, making them less generalizable when new data or domains are encountered. Furthermore, SS-based approaches usually design a single domain-shared structure to learn shared knowledge. However, the relationship between domains can be complex, and shared knowledge that is beneficial to one domain may not be suitable for another domain. Inputting the same knowledge into each domain without filtering will lead to severe performance degradation, which is known as Negative Transfer Problems (NTP) \cite{zhang2022survey, standley2020tasks-multi-task, park2024pacer-cross-domain-rec,li2024aiming-cross-domain-rec}. In addition, some methods only consider raw data features to generate weight vectors or network parameters that control the model, ignoring the fact that the actual learning ability of the model differs from domain to domain, and such inaccuracy may lead to incorrect direction of gradient descent, which results in sub-optimal performance.

\begin{figure}[t]
\includegraphics[width=1\linewidth]{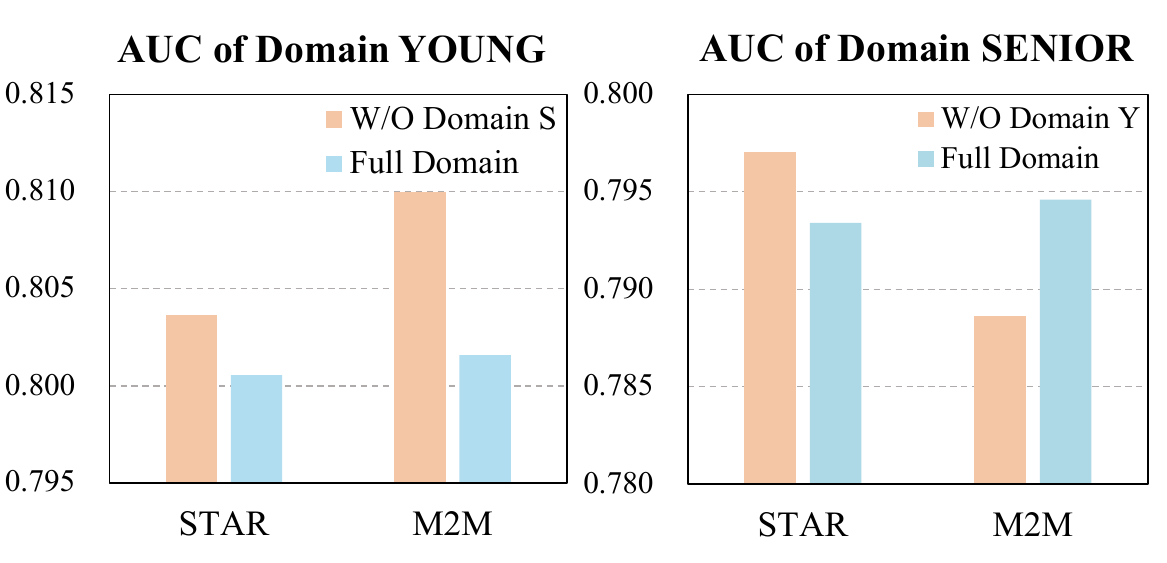}
\caption{The performance of STAR and M2M with selected different domains on Movielens Dataset. 
 ``S'' and ``Y'' are abbreviations for the SENIOR and YOUNG, respectively.}
\label{motivation}
\end{figure}

In order to experimentally validate the above issues on the multi-domain field, we report the domain transfer experiment performance of two classical multi-domain algorithms STAR \cite{sheng2021star-multi-domain-rec} and M2M \cite{zhang2022m2m-multi-domain-multi-task-rec} on the Movielens dataset. Specifically, we report the performance variations of each domain with different selected domains. Movielens dataset describes people's preferences for movies, and the user characteristic $`$age' is utilized to categorize the dataset into three different domains: YOUNG, MIDDLE, and SENIOR. As can be seen  in Figure \ref{motivation}, we can get three key conclusions: i) \uline{The negative transfer problem between domains exists}. When introducing data from domain SENIOR, there is a significant negative transfer of performance in domain YOUNG. This result is not difficult to understand, as teenagers' preferences for films and older adults are usually quite different. ii) \uline{Relationships between domains are not always symmetrical}. In terms of  M2M, domain SENIOR and domain YOUNG show opposite transfer effects to each other, which inspires us that the relationships between domains are not symmetrical, i.e., a positive transfer phenomenon from domain A to domain B does not indicate that domain B has a positive effect to domain A.  iii) \uline{The relationship between domains varies with the learning ability of the model}. For domain SENIOR, when domain YOUNG is introduced, the two models show opposite transfer effects.

To tackle these limitations, we propose a dynamic Similar Domain Selection Principle (SDSP). SDSP presents the initial exploration of selecting suitable domain knowledge for each domain to alleviate NTP. Specifically, we propose a prototype-based domain distance measure to effectively model the complexity relationship between domains. Thereafter, SDSP can dynamically find similar domains for each domain based on the supervised signals of the domain metrics and the unsupervised distance measure from the learned domain prototype. We emphasize that SDSP is a lightweight method that can be incorporated with most existing MDR methods for better performance while not introducing excessive time overheads. To the best of our knowledge, it is the first solution that can explicitly measure domain-level gaps and dynamically select appropriate domains in the MDR field. Moreover, we theoretically analyze the reasons for negative transfer of the multi-domain approachs. Our contributions are summarized as follows:
\begin{itemize} [leftmargin=*]
\item We propose a novel prototype-based domain distance measure that can efficiently model the complex relationship between domains without additional feature engineering.
    \item We introduce a domain selection principle (SDSP), leveraging both the supervised signals and the unsupervised distance measure to select beneficial domains. Besides, SDSP can be integrated into existing MDR methods to improve their performance.
    \item Extensive experiments on three datasets demonstrate the effectiveness and efficiency of our proposed method.
\end{itemize}

\section{Problem Statement}

\subsection{Multi-Domain Recommendation}

In this paper, we focus on multi-domain recommendations for predicting Click-Through Rate (CTR) tasks. The aim of the CTR task is to predict the probability that a user will click on an item based on data from multiple domains. Denote $\mathcal{D}=\left\{1, 2, \dots, D\right\}$ as $D$ distinct domains. The goal is to learn a function $f^d\left(\cdot\right)$ that outputs the predicted CTR rate for each domain $d$:\ $ \hat{y}^d=f^d(\mathbf{x}^d; \Theta^d)$, where $\hat{y}^d \in[0,1]$ denotes the predicted click probability in domain $d$, $\mathbf{x}^d$ denotes the data features of domain $d$ which may include: user features, item features, text features, etc, and  $\Theta^d$ is the network parameter for domain $d$.

Given a training dataset $\mathcal{T}_d=\{(\mathbf{x}_i,y^d_i)\}_{i=1}^{n_d}$ for domain $d$, where $n_d$ is the number of data in domain $d$ and $y^d_i \in\{0,1\}$ is the true labels (1 for clicks and 0 for no clicks). To minimize the loss, the binary cross-entropy loss is usually adopted. Specifically, the CTR loss can be formulated as:
\begin{equation}
\label{cross-entropy}
  \mathcal{L}_{ctr}= - \sum_{d=1}^D \sum_{i=1}^{n_d}\left[y^d_i \log \hat{y}^d_i+\left(1-y_i^d\right) \log \left(1-\hat{y}^d_i\right)\right]
\end{equation}

\subsection{Domain Selection}
Shared network structure, an effective component for transferring knowledge from different domains, is widely used in numerous MDR approaches. For example, the expert networks in MMOE \cite{ma2018mmoe-multi-task-rec}, domain-shared experts in PLE \cite{tang2020ple-multi-task-rec}, and expert generators in M2M \cite{zhang2022m2m-multi-domain-multi-task-rec}. We reflect this domain dependency in the parameter $\Theta$, denoted as $\Theta_\mathcal{D}$, which means the parameter $\Theta$ would be influenced by the data in all domains $\mathcal{D}$. 

While such a structure's effectiveness has been demonstrated, most existing MDR approaches employ a single shared structure to transfer knowledge across all domains. However, the beneficial transfer information should vary across different domains, and unselectively leveraging information from all domains will lead to the NTP. The domain selection strategy is proposed to decouple the single shared structure. Specifically, domain selection tries to find a beneficial domain set for each domain to avoid the NTP caused by large domain gaps or noise information. Without loss of generality, taking domain $d$ as an example, we would like to find the most suitable subset of domains $\mathcal{S}^d \subseteq \mathcal{D}$ ($d \in \mathcal{S}^d$) for it, which will improve the performance of domain $d$ as much as possible while reducing the negative transfer of other domains to it.

\begin{figure*}[t] 
\centering 
\includegraphics[width=0.98\textwidth]{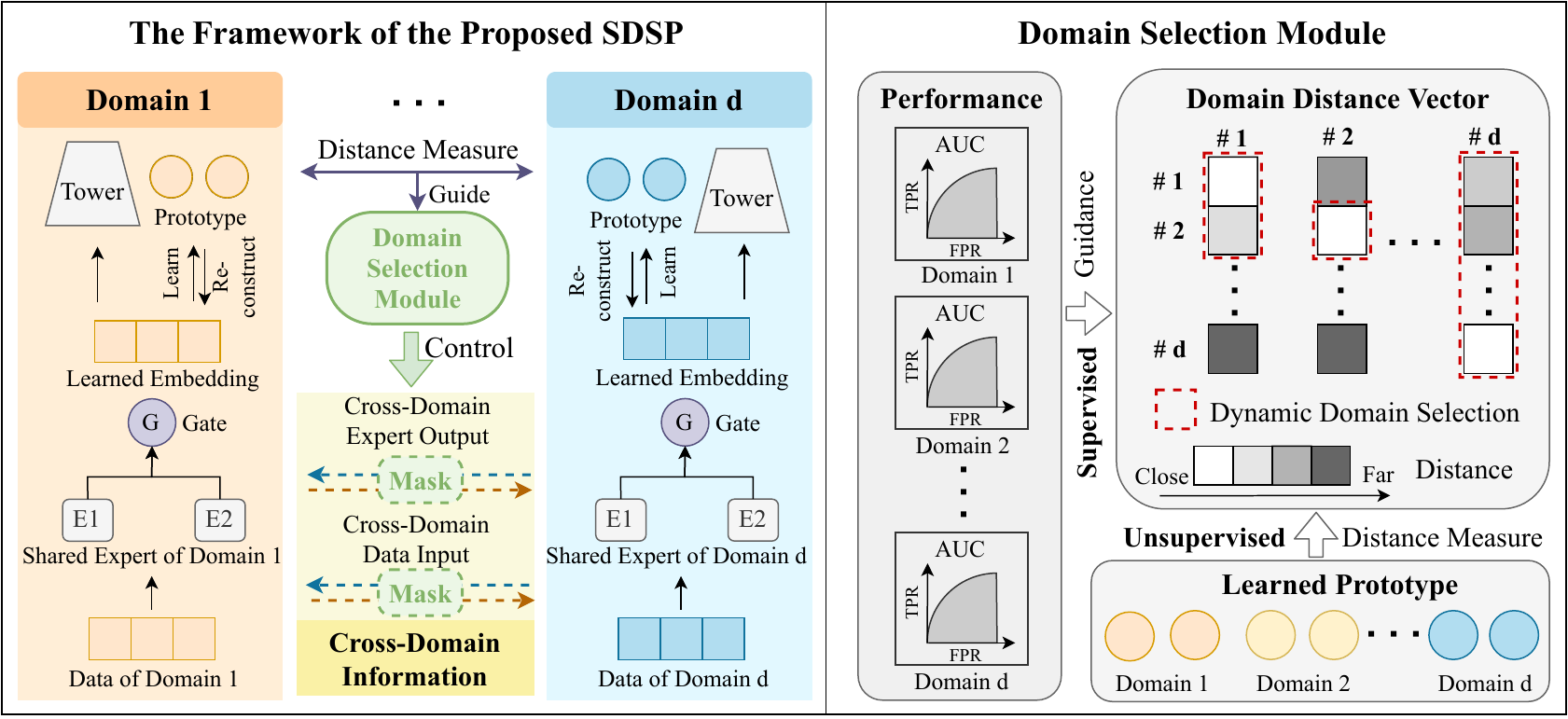} 
\caption{The framework of the proposed SDSP. } 
\label{framework} 
\vspace{-1em}
\end{figure*}

\section{Method}
\subsection{Motivation}
As mentioned above, we attempted to use the domain selection strategy to alleviate the negative transfer problem (NTR) in the MDR field. However, addressing the domain selection problem is not an easy task. 
Firstly, finding the optimal solution among all combinations has been proven to be an NP-hard problem \cite{standley2020tasks-multi-task}, and manually selecting similar domains as hyper-parameters incurs additional time overheads. Besides, there is no prior knowledge to model the complex relationships between domains, so it is highly necessary to propose an effective domain selection algorithm.

As discussed in Section \ref{select}, the current selection methods proposed often adopt the greedy-based algorithm and require trial runs on additional combinations, which imposes huge time overhead and is not suitable in time-sensitive recommendation fields because excessive latency can lead to user attrition. Recently, Bai et al. \cite{bai2022saliency-multi-task}  novelly utilize the gradient of tasks to model task relationships. Similarly, some recent methods \cite{luo2024one-AREAD-multi-domain-rec, sun2020learning,dom_sim} attempt to use pre-cluster to decouple the domain relations. However, both the above methods depend on the simplified relation modeling assumption, which ignores the asymmetry of domain relationships in real scenarios \cite{sun2020learning}. As we illustrate in the INTRODUCTION section, domain A has a negative transfer impact on domain B, but it does not always indicate that domain B has a negative transfer impact on domain A.

To tackle these issues, we proposed a similar domain principle (SDSP) for Multi-Domain Recommendation. As shown in Figure \ref{framework}, our proposed SDSP consists of three components: flexible domain representation learning, prototype-based domain distance measure, and dynamic similar domain selection principle.

\subsection{Flexible Domain Representation Learning}
 Without loss of generality, we adopt a basic Multi-gate Mixture-of-Experts (MMOE) structure as an example to learn domain-specific representations in this paper. Because numerous MDR approaches adopt such a paradigm \cite{tang2020ple-multi-task-rec,shen2021sarnet} as a shared structure to transfer knowledge between domains. Formally, the output for domain $d$ is defined as:
\begin{equation}
    \boldsymbol{y}^d={tower}^d\left(\boldsymbol{h}^d\right), 
\end{equation}
where  ${tower}^d(\cdot)$ is the domain-specific tower network. $\boldsymbol{h}^d$ is learned embedding in domain $d$. Besides,  $\boldsymbol{h}^d$ is calculated by:
\begin{align}
\label{represent}
        & \boldsymbol{h}^d =   \boldsymbol{g}^d \boldsymbol{e}(\boldsymbol{x}^d), \\
    & \boldsymbol{g}^d = \text{softmax}(\boldsymbol{W}_{g}^d\boldsymbol{x}^d), \label{gate}
\end{align}
where $\boldsymbol{x}^d$ is the input representation in domain $d$,  $\boldsymbol{e}(\cdot)$ is the representation set of the whole expert network, $\boldsymbol{g}^d \in \mathbb{R}^{E }$ is the weight generated by the gating network for domain $d$, and  $E = \sum_{i = 1}^D E_d$, where $E_d$ is the number of expert in domain $d$.  

Slightly different from MMOE, we decouple the expert network shared by all domains and design a specialized shared expert network for each domain. This partitioning process is somewhat similar to PLE's specific expert network, but we only decouple the shared expert network part. Specifically, 
\begin{equation}
\label{expert}
    \boldsymbol{e}(\boldsymbol{x}^d) = \left[ \boldsymbol{e}_i \right]_{i=1}^E= \left[ \boldsymbol{e}_1^1, ... ,\boldsymbol{e}_{E_1}^1, ... , \boldsymbol{e}_1^D, ... ,\boldsymbol{e}_{E_D}^D  \right],
\end{equation}
where $\boldsymbol{e}_j^i$  denotes the $j$-th expert network in $i$-th domain, and $E_d$ is the number of expert in domain $d$.

No additional performance-enhancing mechanisms are employed to highlight the effectiveness of the domain selection strategy. In addition, we emphasize that the above learning approach is flexible and can be replaced with most existing MDR architectures, thereby achieving a better performance.

As mentioned before,  blindly using information from all domains will result in serious NTP. In the later subsection, we will introduce how to use the domain selection mechanism to alleviate the NTP.

\subsection{Prototype-based Domain Distance Measure}
\label{proto_measure}
To measure the distance between domains and consider the learning ability of the model, an intuitive approach is to compute the similarity between domain-specific representations, i.e., calculating $dis(\boldsymbol{h}^{d}, \boldsymbol{h}^{d'})$, where $\boldsymbol{h}^d$ is the learned representation in $d$-th domain. However, such an approach would incur a huge time overhead of $\mathcal{O}(N^2)$, where $N$ is the number of samples. Even in the case of batch processing, the time complexity is still $\mathcal{O}(B^2)$, where $B$ is the batch size. This is unacceptable for the latency-sensitive recommendation field. In addition, unselectively measuring all data will inevitably be interfered with by low-quality or noisy data.

In order to tackle these issues, we introduce the concept of prototype learning to measure the distance between domains. The definition of the prototype is ``a representative embedding for a group of instances'' \cite{liprototypical}. With the learned prototype, we can effectively reduce the time overhead and alleviate the interference of low-quality or noisy data \cite{gan2024peace}. Specifically, we devise a prototype learning encoder to learn the fine-grained domain-level prototype: 
\begin{equation}
    \boldsymbol{p}^d = P-Encoder(\boldsymbol{h}^d),
\end{equation}
where $\boldsymbol{p}^d \in \mathbb{R}^{M \times D}$ is learned prototype in domain $d$, and $\boldsymbol{h}^d \in \mathbb{R}^{B \times D}$ is the learned representation in $d$-th domain defined in Eq. (\ref{represent}), with the batch training setup. $M$, $D$, and $B$ are the number of prototypes, feature dimension, and batch size, respectively. 

Besides, we devise a prototype decoder and a reconstruction loss to ensure the high quality of the learned prototype. Specifically, we consider the principle that the learned prototype is sufficiently representative if it can effectively decode the original embedding.
\begin{align}
    & \hat{\boldsymbol{h}}^d = P-Decoder(\boldsymbol{p}^d) \\
    & \mathcal{L}_{rec} = \sum_{d=1}^D \Vert \boldsymbol{h}^d - \hat{\boldsymbol{h}}^d \Vert_2^2 \label{rec_loss}
\end{align}
where $\hat{\boldsymbol{h}}^d$ is the recontructed sample embedding in domain $d$.

After obtaining the domain-level prototype \cite{ysj, wxh1}, the next step is to compute the domain-to-domain distance. Referring to the conception of segmented-path distance defined in the cluster measure field~\cite{besse2016review}, we propose the following prototype-based domain distance measure: 
\begin{definition} (\textbf{Distance between Prototype and Domain})
The distance from the prototype $\boldsymbol{p}$  to the domain $d$ is defined as the distance from the prototype $\boldsymbol{p}$ to the nearest prototype in the domain $d$, i.e., 
$$dis(\boldsymbol{p}, d) = \min_{i \in \{1, 2, \dots, M\}} dis(\boldsymbol{p}, \boldsymbol{p}^{d}_i),$$
where $\boldsymbol{p}^{d}_i$ represents the $i$-th prototype in domain $d$, and $M$ is the number of prototypes. In our experiment, the distance $dis(\boldsymbol{p}, \boldsymbol{p}^{d}_i)$ is calculated by the L2 norm between the $\boldsymbol{p}$ and $\boldsymbol{p}^{d}_i$.
\end{definition}

\begin{definition} (\textbf{Distance between Domain and Domain})
The distance from domain $d_1$ to domain $d_2$ is defined as the average of the distances from all prototypes in domain $d_1$ to domain $d_2$, i.e., $$dis(d_1, d_2) = \frac{1}{M}\sum_{i=1}^M dis(\boldsymbol{p}^{d_1}_i, d_2). $$
\end{definition}

\begin{remark}
Although the defined measure seems simple, we emphasize its superiority as follows: 
\begin{itemize}[leftmargin=*]
\item \textbf{Efficient.} \ With the learned prototype, the time overhead of the domain measure part would decrease effectively from $\mathcal{O}(B^2)$ to $\mathcal{O}(BM+M^2)$, where   $B$, $M$ is the number of batch size and prototype. In our setting, $M \ll B$, thus the prototype-based measure can effectively save the time overhead of the algorithm.
\item \textbf{Asymmetrical.} \ The distance defined is asymmetric. Specifically, $dis(d_1, d_2) \neq dis( d_2, d_1)$ in most settings,  which satisfies the need for multi-domain relationship modeling.
As mentioned in the INTRODUCTION section, the relationships between domains are not symmetrical, i.e., a positive transfer phenomenon from domain A to domain B does not indicate that domain B has a positive transfer effect to domain A. 
\item \textbf{Model-Aware.} \ Since the prototype is learned from the model's representation rather than the raw data, the prototype can be better aware of the model's current learning capacity and thus adaptively measure the differences between domains.
\end{itemize}
\end{remark}

\begin{remark}
    With the calculated domain distance, the combination of similar domains we need to search is drastically reduced from $\mathcal{O}(2^{D-1})$ to $\mathcal{O}(D)$,  where $D$ is the number of domains. Taking the Movielens dataset as an example, if the distance from other domains to domain YOUNG is known to be sorted as YOUNG < MIDDLE < SENIOR, there are only three states that we need to discuss for similar domains of domain YOUNG, i.e., $[Y]$, $[Y, M]$, and $[Y, M, S]$. The more the number of domains, the more the number of states is reduced.
\end{remark}
\begin{algorithm}[t]
\caption{Similar Domain Selection Principle (SDSP)}
\begin{algorithmic}[1]
    \REQUIRE Decay Rate $d\_r$ 
    \ENSURE Similar Domain Set $ \{ \mathcal{S}^d \}_{d=1}^D$, Prediction Click Label $\hat{y}$  
\STATE \textbf{Initialize:} Probability $ p = 1 $, Selection Loop $l = 2$ 
      \WHILE{not converged}
\STATE Training model with training data $\mathcal{T}$ and learn prototype $\boldsymbol{p}$
\IF{$iter$ \% $l$ == 0}
\STATE Measure domain distance with learned prototype $\boldsymbol{p}$
\STATE Calculate metric $\{\boldsymbol{r^d}\}_{d=1}^D$ with validation data $\mathcal{V}$
\FOR{d = 1 to D}
\STATE Update value function of state $\mathcal{S}^d$ using metric $\boldsymbol{r}^d$
\STATE Generate a random value $rnd \in [0, 1] $
        \IF{$ rnd \leq p $}
            \STATE Randomly select a similar domain set for $ \mathcal{S}^d $
        \ELSE
            \STATE Select best domain set for $ \mathcal{S}^d $  based on value function
        \ENDIF
\ENDFOR
\STATE Update masks $ \{ \boldsymbol{M}^d \}_{d=1}^D$ using similar domain $ \{ \mathcal{S}^d \}_{d=1}^D$ 
\STATE Update probability $p \leftarrow p \times d\_r$ 
\ENDIF
\STATE Update iteration $iter \leftarrow iter + 1$ 
\ENDWHILE
\end{algorithmic}
\end{algorithm}

\subsection{Dynamic Similar Domain Selection}
Although the number of states to be considered is significantly reduced after the intra-domain distance measure, it is still a difficult problem to select its beneficial domains for each domain because there is no prior knowledge for us to model inter-domain relationships. To address these issues, we model the domain selection problem as a Multi-Armed Bandit Problem and use the classical Epsilon-Greedy (EG) algorithm to solve it. 

EG is a classical selection algorithm widely used in deep learning, which combines random and greedy algorithms to deal with exploration and exploitation dilemmas. The main idea of EG is to control the usage of greedy or random algorithms by a small probability. Specifically, the Epsilon-greedy method selects an action at a decision point by first generating a random value $rnd$ in the interval [0,1], and if $rnd$ is greater than Epsilon, the Epsilon-greedy method uses the greedy algorithm to select the action. Otherwise, it uses the random algorithm to select the action \cite{liu2022understanding}.

We add a decay mechanism to the Epsilon-greedy algorithm, expecting the model to explore as many new states as possible in the early stages of training while focusing more on learning the optimal states that have been found in the later stages of training. State Space $\mathcal{S}$ is a discrete state space, which consists of each combination of final domain selection. The value function for each domain is defined as the domain metric (e.g., AUC).

After obtaining the similar domain set $\mathcal{S}^d$ for each domain $d$, we can revise Eq. (\ref{represent}) through the mask mechanism to avoid the negative transfer effect across domains:
 \begin{align}
     & \boldsymbol{h}^d_{revised} = \boldsymbol{g}_{mask}^d(\boldsymbol{x}^d) \boldsymbol{e}(\boldsymbol{x}^d) \\
     &\boldsymbol{g}_{mask}^d(\boldsymbol{x}^d)= \text{softmax} (\boldsymbol{g}^d + \mathcal{M}^d) 
 \end{align}
where $\boldsymbol{g}^d$ is calculated by Eq. (\ref{gate}),  and $\mathcal{M}^d \in \mathbb{R}^{E}$ is the $d$-th domain's mask and is defined by
\begin{equation}
    \mathcal{M}^d_i  = \begin{cases} 0, & \text { if } d'_i \in \mathcal{S}^d, \\  -\infty, & \text { otherwise. }\end{cases}
\end{equation}
where $\mathcal{M}^d_i$ is the $i$-th value of the $\mathcal{M}^d$ and $d'_i$ is the corresponding domain of $\boldsymbol{e}_i$ defined in Eq. (\ref{expert}).

\subsection{Overall Loss Function}
The overall optimization goal of our model is:
\begin{equation}
    \mathcal{L}_{final} = \mathcal{L}_{ctr} + \gamma \mathcal{L}_{rec}
\end{equation}
where $\mathcal{L}_{ctr}$ is the cross-entropy loss function and $\gamma$ is the prototype learning hyper-parameter.  $\mathcal{L}_{rec}$ is defined by Eq. (\ref{rec_loss}),  and $\mathcal{L}_{ctr}$ can be calculated by $- \sum_{d,i}[y^d_i \log \hat{y}^d_i+(1-y_i^d) \log (1-\hat{y}^d_i)]$, $\hat{y}_i^d$ and $y_i^d$ denote the prediction and the ground truth in the $d$-th domain, respectively.

\section{Theoretical Analysis}
In this section, we attempt to theoretically analyze the reasons for negative inter-domain transfer. 

For convenience, we give the following notations: $\mathcal{X} = \{\mathbf{X}^1, \ldots \mathbf{X}^D\}$ represents multi-domain data, where $\mathbf{X}^i$ contains the user information and the corresponding item or contextual information in the $i$-th domain, $D$ is the number of domain. $\mathcal{Y} = \{\mathbf{Y}^1, \ldots \mathbf{Y}^D\}$  represents the interaction record set in different domains. $\mathcal{D} = \{1, \ldots, D\}$ is the domain number set. To simplify the following deduction, we assume $(\mathcal{X}, \mathcal{Y})$ follows the underlying joint distribution $p$, which is a common assumption of theorem analysis. 

Considering two random domains data $\mathbf{X}^i$ and $\mathbf{X}^j$, $i \neq j$, we define $I(\mathbf{X}^i ; \mathbf{X}^j)$ to be the shannon mutual information between $\mathbf{X}^i$ and $\mathbf{X}^j$. Similarly, we use $H(\mathbf{Y}^i \mid \mathbf{X}^i, \mathbf{X}^j)$ to denote the conditional entropy of $\mathbf{Y}^i$ given the two modalities as input. Following common practice, for recommendation tasks, $\ell_{\mathrm{CE}}(\hat{y}, y)$ is the cross-entropy loss between the prediction $\hat{y}$ and the ground-truth transaction $y$. 

\begin{assumption}
\textbf{\text{(Optimality)}} The domain contributes the most to the $i$-th domain's interaction is exactly $i$-th, i.e., $i = \arg \max_jI(\mathbf{X}^j; \mathbf{Y}^i)$. 
\end{assumption}
\begin{remark}
   The above assumption happens frequently in real scenarios. Without the target domain's information, it would be extremely hard to predict its interaction. 
\end{remark}

\begin{theorem}
\label{select_theorem}
Denote $\mathbf{Z}^i$ is the $i$-th domain's embedding learned from the original multi-domain data $\mathcal{X}$. If there exist a subset $\mathcal{S}^i \subseteq \mathcal{D}$ such that
$\mathbf{Z}^i = g^j(\mathbf{X}^j)$, $ \forall j \in \mathcal{S}^i$, where $g^j$ can be any function. With the optimality assumption, we have the following conclusion: 
$$
 \inf _h \mathbb{E}_p[\ell_{\mathrm{CE}}(h(\mathbf{Z}^i), \mathbf{Y}_i)]-\inf _{h^{\prime}} \mathbb{E}_p[\ell_{\mathrm{CE}}(h^{\prime}(\mathbf{X}^1,\ldots, \mathbf{X}^D), \mathbf{Y}_i)] \geq \Delta_p^i,
$$ 
where  $\Delta_p^i = I(X_i; Y_i)-\min_{j \in \mathcal{S}^i} I(X^j; Y_i)$, $ i \in \{1, \ldots, D\}$
\end{theorem}
The term $\inf _h \mathbb{E}_p[\ell_{\mathrm{CE}}(h(\cdot), \mathbf{Y}_i)]$ reflects the best function $f$ that could be found to minimize the generalization cross-entropy loss in the $i$-th domains. Due to the space limit, the proof of  Theorem \ref{select_theorem}  would be shown in Appendix \ref{appendix_proof}. 

\begin{remark}
From the above theorem, we have the following insights: The lower bound of the information gap depends on the domain that contributes the least to the $i$-th domain. If $Z^i$ learns too much meaningless information from domains of poorer quality, it can lead to serious negative transfer, which inspires us that decoupling the shared structure with less informative domains could better decrease the domain gap, thereby alleviating the negative transfer problem.
\end{remark}
\section{Experiments}
In this section, we conduct extensive experiments to verify the effectiveness of the proposed method. The experiments are designed to answer the following questions:

\noindent
\textbf{(RQ1):} Does the performance of proposed SDSP exceed the exsiting state-of-the-art methods?

\noindent
\textbf{(RQ2):} Can the proposed selection principle apply other MDR model to improve their performance?

\noindent
\textbf{(RQ3):} How do hyper-parameters influence the proposed model?

\noindent
\textbf{(RQ4):} Do the proposed two sub-modules contribute to the AUC?

\noindent
\textbf{(RQ5):} Is the proposed SDSP helpful in alleviating the NTP?

\noindent
\textbf{(RQ6):} Is the SDSP an efficient module?

\begin{table}[t]
\caption{The statistics of the evaluation datasets.}
\vspace{-1em}
\label{dataset}
\begin{tabular}{ccccc}
\toprule
Dataset   & \#Domain & \#Users & \# Items & \#Interactions \\ \midrule
Movielens & 3        & 6,040    & 10,413    & 800,167         \\
Amazon    & 3        & 100,345  & 53,665    & 658,827         \\
Douban    & 3        & 6,671    & 184,619   & 1,348,399   \\ \bottomrule
\end{tabular}
\end{table}

\begin{table*}[t]
\fontsize{8}{9.7}\selectfont 
\centering
\caption{Empirical evaluation and comparison of SDSP with ten baseline methods on three datasets. $``*"$ indicates the statistically significant
improvements (i.e., two-sided t-test with $p < 0.05$) over the best baseline. $\uparrow$ : higher is better; $\downarrow:$ lower is better.}
\label{main_exp}
\resizebox{\linewidth}{!}{
\renewcommand\arraystretch{1}
\tabcolsep=0.1cm
\scalebox{1.07}{
\begin{tabular}{cccccccccc||cccccc} \toprule
\multirow{3}{*}{Method} & \multicolumn{9}{c||}{Performance for Each   Domain (AUC $\uparrow$)}                                                                                             & \multicolumn{6}{c}{Overall Performance}                                                                                                                                              \\ \cmidrule{2-16} 
    & \multicolumn{3}{c}{Movielens}                       & \multicolumn{3}{c}{Amazon}                          & \multicolumn{3}{c||}{Douban}                         & \multicolumn{2}{c}{Movielens}                              & \multicolumn{2}{c}{Amazon}                                 & \multicolumn{2}{c}{Douban}                                 \\ \cmidrule{2-16} 
                 & D1              & D2              & D3              & D1              & D2              & D3              & D1              & D2              & D3              & AUC $\uparrow$  & \multicolumn{1}{l}{Logloss $\downarrow$} & AUC $\uparrow$  & \multicolumn{1}{l}{Logloss $\downarrow$} & AUC $\uparrow$  & \multicolumn{1}{l}{Logloss $\downarrow$} \\ \midrule
SharedBottom & 0.8082 & 0.8055 & 0.7954 & 0.6998 & 0.6782 & 0.7322 & 0.7130 & \textbf{0.7366} & 0.8117 & 0.8024 & 0.5348 & 0.7058 & 0.4910 & 0.7992 & 0.5179 \\
MMOE         & 0.8114 & 0.8122 & 0.7998 & 0.7005 & 0.6826 & 0.7323 & 0.7082 & 0.7250 & 0.8109 & 0.8077 & 0.5238 & 0.7073 & 0.4915 & 0.7974 & 0.5200 \\
PLE          & 0.8105 & 0.8134 & 0.8014 & 0.7025 & 0.6796 & 0.7324 & 0.7096 & 0.7305 & 0.8104 & 0.8085 & 0.5230 & 0.7075 & 0.4910 & 0.7976 & 0.5202 \\
STAR         & 0.8005 & 0.8041 & 0.7934 & 0.6739 & 0.6462 & 0.6969 & 0.7080 & 0.7237 & 0.8083 & 0.7996 & 0.5344 & 0.6747 & 0.4980 & 0.7952 & 0.5225 \\
SAR-Net      & 0.8123 & 0.8134 & 0.8004 & 0.6978 & 0.6785 & 0.7239 & 0.7163 & 0.7345 & 0.8117 & 0.8087 & 0.5244 & 0.7020 & 0.4954 & 0.7991 & 0.5196                                  \\
Adasparse               & 0.8085          & 0.8114          & 0.8012          & 0.6839          & 0.6640          & 0.7100          & 0.7067          & 0.7304          & 0.8113          & 0.8073          & 0.5230                                   & 0.6895          & 0.4791                                   & 0.7981          & 0.5201                                   \\
Adaptdhm                & 0.8012          & 0.8092          & 0.7963          & 0.6982          & 0.6807          & 0.7285          & 0.7114          & 0.7235          & 0.8139          & 0.8030          & 0.5266                                   & 0.7056          & 0.4802                                   & 0.7998          & 0.5193                                   \\
M2M                     & 0.8016          & 0.8080          & 0.7946          & 0.6837          & 0.6656          & 0.7134          & 0.7027          & 0.7147          & 0.8151          & 0.8017          & 0.5297                                   & 0.6915          & 0.4887                                   & 0.7983          & 0.5230                                   \\
PPNet                   & 0.8080          & 0.8107          & 0.7982          & 0.6768          & 0.6492          & 0.6988          & 0.7122          & 0.7327          & 0.8110          & 0.8056          & 0.5255                                   & 0.6779          & 0.4765                                   & 0.7984          & 0.5188                                   \\
EPNet                   & 0.8040          & 0.8107          & 0.7979          & 0.7009          & 0.6816          & 0.7305          & 0.7132          & 0.7250          & 0.8134          & 0.8049          & 0.5249                                   & 0.7078          & 0.4746                                   & 0.8000          & 0.5180                                   \\ \midrule
MMOE+SDSP               & \textbf{0.8197}$^*$ & \textbf{0.8191}$^*$ & \uline{0.8072}$^*$ & \textbf{0.7056}$^*$ & \uline{0.6841}$^*$ & \textbf{0.7351}$^*$ & \textbf{0.7168}$^*$ & 0.7308 & \uline{0.8155}$^*$       & \textbf{0.8151}$^*$ & \textbf{0.5143}$^*$                          & \uline{0.7107}$^*$ & \uline{0.4680}$^*$                    & \uline{0.8024}$^*$ & \uline{0.5158}$^*$                    \\
PLE+SDSP                & \uline{0.8147}$^*$    & \uline{0.8190}$^*$    & \textbf{0.8090}$^*$       & \uline{0.7049}$^*$    & \textbf{0.6848}$^*$       & \uline{0.7342}$^*$    & \textbf{0.7168}$^*$ & \uline{0.7322}      & \textbf{0.8165}$^*$             & \uline{0.8147}$^*$    & \uline{0.5151}$^*$                             & \textbf{0.7116}$^*$       & \textbf{0.4652}$^*$                          & \textbf{0.8031}$^*$       & \textbf{0.5147}$^*$          \\ \bottomrule                       
\end{tabular}}}
\end{table*}

\begin{table*}[t]
\caption{Compatibility Experiment Results. $``*"$ indicates the statistically significant
improvements (i.e., two-sided t-test with $p < 0.05$) over the original baseline.}
\label{Compatibility}
\fontsize{8}{9.7}\selectfont 
\begin{tabular}{ccccccccccccc}
\toprule
\multirow{2}{*}{ AUC   ($\uparrow$)} & \multicolumn{4}{c}{Movielens}      & \multicolumn{4}{c}{Amazon}                        & \multicolumn{4}{c}{Douban}                                            \\ \cmidrule{2-13}
  & D1              & D2              & D3              & Overall         & D1              & D2              & D3              & Overall         & D1              & D2              & D3              & Overall         \\  \midrule
MMOE & 0.8114 & 0.8122 & 0.7998 & 0.8077 & 0.7005 & 0.6826 & 0.7323 & 0.7073 & 0.7082 & 0.7250 & 0.8109 & 0.7974          \\
MMOE+SDSP                                 & \textbf{0.8197}$^*$ & \textbf{0.8191}$^*$ & \textbf{0.8072}$^*$ & \textbf{0.8151}$^*$ & \textbf{0.7056}$^*$ & \textbf{0.6841}$^*$ & \textbf{0.7351}$^*$ & \textbf{0.7107}$^*$ & \textbf{0.7168}$^*$ & \textbf{0.7308}$^*$ & \textbf{0.8155}$^*$ & \textbf{0.8024}$^*$ \\ \midrule
PLE & 0.8105 & 0.8134 & 0.8014 & 0.8085 & 0.7025 & 0.6796 & 0.7324 & 0.7075 & 0.7096 & 0.7305 & 0.8104 & 0.7976       \\
PLE+SDSP                                  &  \textbf{0.8147}$^*$ & \textbf{0.8190}$^*$ & \textbf{0.8090}$^*$ & \textbf{0.8147}$^*$ & \textbf{0.7049}$^*$ & \textbf{0.6848}$^*$ & \textbf{0.7342}$^*$ & \textbf{0.7116}$^*$ & \textbf{0.7168}$^*$ & \textbf{0.7322}$^*$ & \textbf{0.8165}$^*$ & \textbf{0.8031}$^*$\\ \midrule
SARNet & 0.8123 & 0.8134 & 0.8004 & 0.8087 & \textbf{0.6978} & 0.6785 & 0.7239 & 0.7020 & 0.7163 & 0.7345 & 0.8117 & 0.7991         \\
SARNet+SDSP                               & \textbf{0.8130}$^*$ & \textbf{0.8159}$^*$ & \textbf{0.8048}$^*$ & \textbf{0.8113}$^*$ & 0.6947 & \textbf{0.6822}$^*$ & \textbf{0.7261}$^*$ & \textbf{0.7047}$^*$ & \textbf{0.7205}$^*$ & \textbf{0.7418}$^*$ & \textbf{0.8137}$^*$ & \textbf{0.8021}$^*$ \\ \bottomrule
\end{tabular}
\end{table*}

\subsection{Experiment Setting}
\subsubsection{\textbf{Datasets.}}
We perform experiments with three public datasets. More statistics are presented in Table \ref{dataset}.

\begin{itemize}[leftmargin=*]
    \item \textbf{Movielens\footnote{\url{https://grouplens.org/datasets/movielens/}}}:  The dataset provides information on the preferences of people about movies, including 7 features related to users and 2 features related to items. The "age" feature of users is utilized to separate the dataset into three domains. Furthermore, the dataset is divided randomly into three parts: training, validation, and testing sets, with an 8:1:1 proportion.
    \item \textbf{Amazon-5core\footnote{\url{https://cseweb.ucsd.edu/~jmcauley/datasets.html##amazon_reviews}}}:  The dataset is a compact subset from Amazon, including users and items with a minimum of 5 interactions. Three domains with overlapped users and items, "Clothing", "Beauty",  and "Health" are utilized for the training and evaluation process. For simplicity, we denote this dataset as "Amazon".
    \item \textbf{Douban\footnote{\url{https://github.com/FengZhu-Joey/GA-DTCDR/tree/main/Data}}}:  The dataset was collected from Douban, which is randomly split into training, validation, and test sets, following an 8:1:1 ratio. The user's ratings (ranging from 1 to 5)  higher than 3  are considered valid clicks.
\end{itemize}

\subsubsection{\textbf{Baselines.}}
To verify the effectiveness of the proposed method, the following baseline is utilized to compare:
\begin{itemize}[leftmargin=*]
    \item \textbf{SharedBottom} \cite{caruana1997sharedbottom-multi-task} is a classical multi-task method that utilizes a common bottom layer to capture the task-shared knowledge. In the multi-domain learning field, each task tower is replaced with a domain-specific tower.
    \item \textbf{MMOE} \cite{ma2018mmoe-multi-task-rec}  employs multiple bottom networks  with distinct gate networks for each task to merge information effectively.
    \item \textbf{PLE} \cite{tang2020ple-multi-task-rec} uses a novel progressive layered extraction structure that separates shared and task-specific components, improving task performance across varied correlations.
    \item \textbf{STAR} \cite{sheng2021star-multi-domain-rec}  uses shared centered parameters and domain-specific parameters to improve click-through rate prediction.
     \item \textbf{SARNet} \cite{shen2021sarnet} leverages scenario-aware attention modules and a multi-scenario gating module to capture cross-scenario interests.
    \item \textbf{Adasparse} \cite{yang2022adasparse_multi-domain-rec} improves multi-domain CTR prediction by learning adaptively sparse structures using domain-aware neuron weighting and sparsity regularization.
        \item \textbf{Adaptdhm} \cite{li2022adaptdhm-multi-domain-rec} employs a hierarchical structure with a dynamic routing-based distribution adaptation module to cluster samples.
    \item \textbf{M2M} \cite{zhang2022m2m-multi-domain-multi-task-rec} utilizes a multi-task multi-scenario meta-learning framework with meta units to model inter-scenario correlations.
    \item \textbf{EPNET} \cite{chang2023pepnet-multi-domain-multi-task-rec} performs personalized embedding selection to fuse features of varying importance across domains.
        \item \textbf{PPNET } \cite{chang2023pepnet-multi-domain-multi-task-rec} dynamically modifies DNN parameters to balance task sparsity, enabling tailored multi-domain recommendations.
        \end{itemize}

\subsubsection{\textbf{Metrics.}}
Following previous work \cite{gao2024hierrec-multi-domain-rec-jingtong, li2024scenario-multi-domain-rec-xiaopeng}, two metrics are utilized to assess the models' performance on the test set: the Area Under the ROC Curve (AUC) and LogLoss. Typically, a higher AUC or a lower LogLoss value indicates superior performance. Notably, according to previous studies \cite{song2019autoint}, an improvement in AUC at the 0.001 level (1\textperthousand) in CTR prediction tasks is considered significant \cite{luo2024one-AREAD-multi-domain-rec} and can lead to substantial commercial benefits online.

\subsubsection{\textbf{Implementation Details.}}
For simplicity, we adopt a single-layer MLP structure for both the prototype encoder and decoder. The number of prototypes is fixed to 10. Specifically, we design a multi-domain sampler to ensure that the amount of data from each domain under each batch is fixed for the purpose of learning prototypes. We searched for the learning rate in the range of [0.01, 0.001,0.0001]. For prototype learning parameter $\alpha$, we search in the range: $[1e^{-3}, 1e^{-4}, 1e^{-5}]$. The number of experts for each domain is fixed to 1. The probability $p$ is initialized to 1, and the decay rate is 0.9. The default parameter setting can be found in Appendix \ref{appendix_default}. To validate the effectiveness of SDSP, we combine the proposed SDSP with the classical MMOE and PLE to report the performance and the subsequent analysis. 
\vspace{-2pt}
\subsection{Performance Comparison (RQ1)}
To evaluate the effectiveness of SDSP, we conduct performance comparison experiments. The experimental results, as presented in Table \ref{main_exp}, allow us to draw the following conclusions:
\begin{itemize}[leftmargin=*]
\item  Compared with multi-task methods (SharedBottom, MMOE, PLE), the proposed method achieves better performance. For example, MMOE+SDSP exceeds SharedBottom by 1.26\%, 3.86\%, and 0.32\% in overall AUC. This improvement highlights the effectiveness of SDSP in mitigating negative transfer caused by inappropriate sharing structure in classic multi-task methods. 

\item Compared with multi-domain methods (STAR, SAR-Net, AdaptDHM), SDSP achieves better performance. Taking SAR-Net as an example, on the Movielens dataset, our method PLE+SDSP achieves an overall AUC of 0.8147, surpassing SAR-Net's 0.8087 by 0.60\%. This demonstrates that our domain selection method can effectively alleviate the inter-domain negative effects by dynamically selecting similar domains for each domain.

\item Adasparse, as a representative sparse MDR method, outperforms most methods on the Movielens dataset. However, Adasparse is still inferior to our method. This suggests that while sparse methods can alleviate negative migration between domains to some extent, relying only on the weights learned from representations is not sufficient to avoid NTP and guidance for domain relationships is also required.
\end{itemize}


\begin{figure*}[t]
    \centering
    \begin{subfigure}{0.24\textwidth}  \label{proto_hyper}\includegraphics[width=\textwidth]{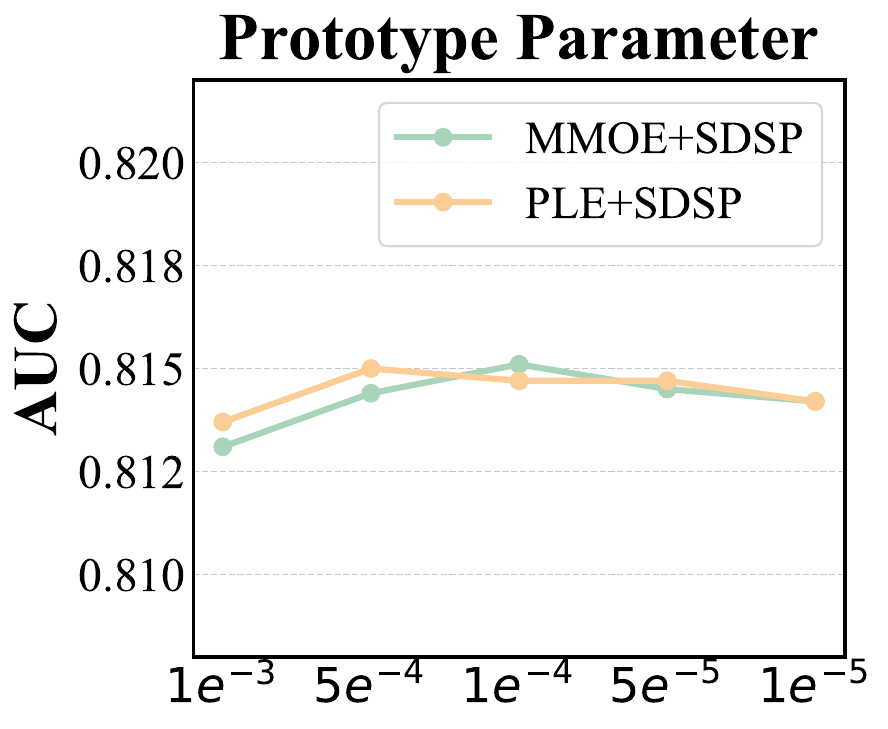}
   \end{subfigure}
    \begin{subfigure}{0.24\textwidth} 
    \label{expert_hyper}{\includegraphics[width=\textwidth]{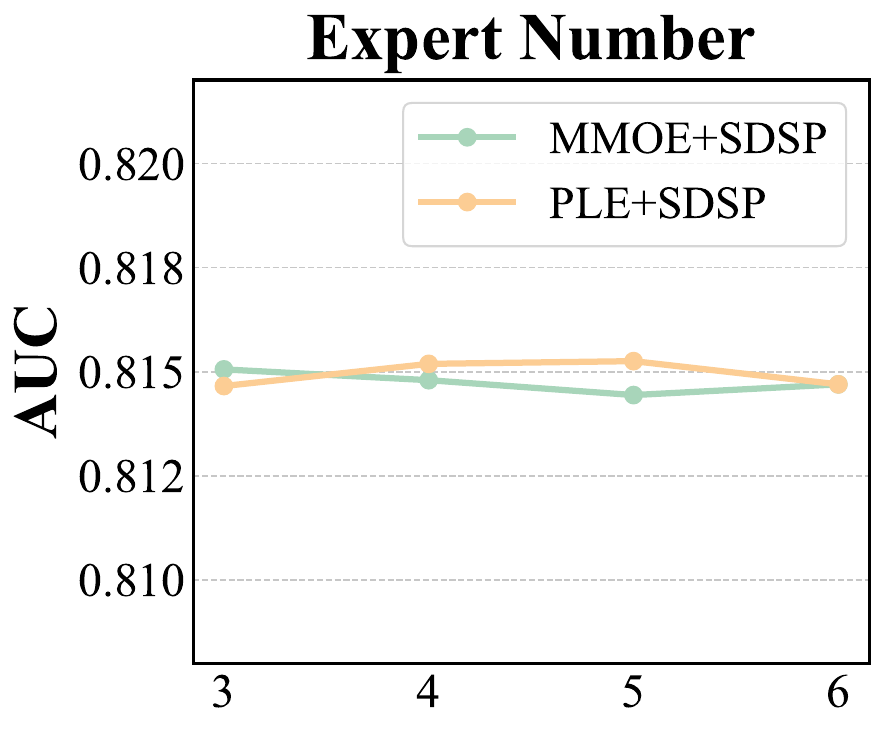}}
       \end{subfigure}
           \begin{subfigure}{0.24\textwidth} 
    \label{batch}{\includegraphics[width=\textwidth]{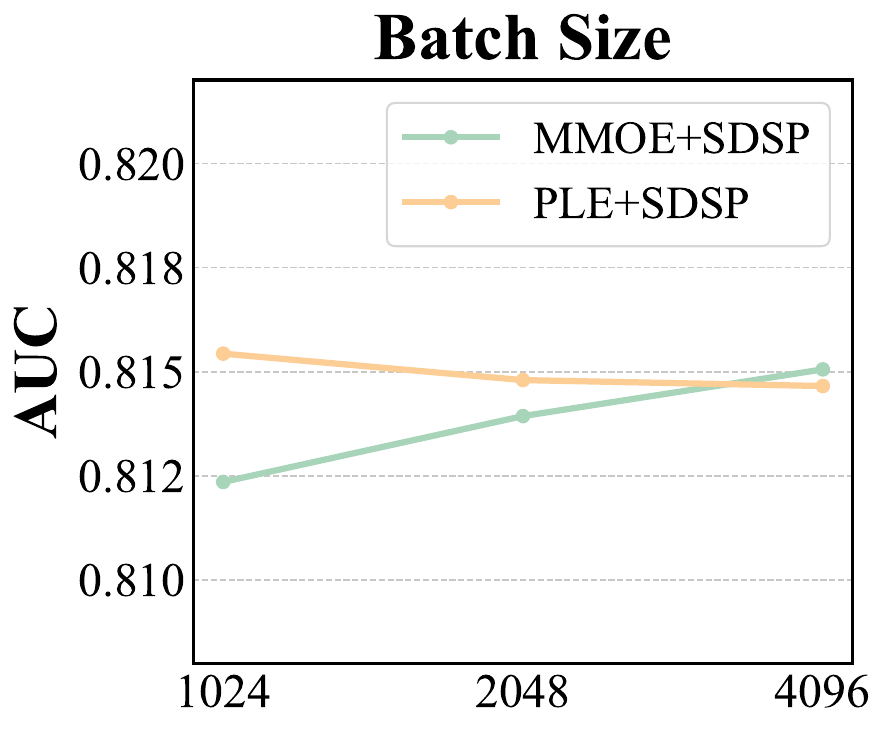}}
       \end{subfigure}
           \begin{subfigure}{0.24\textwidth} 
    \label{proto_num}{\includegraphics[width=\textwidth]{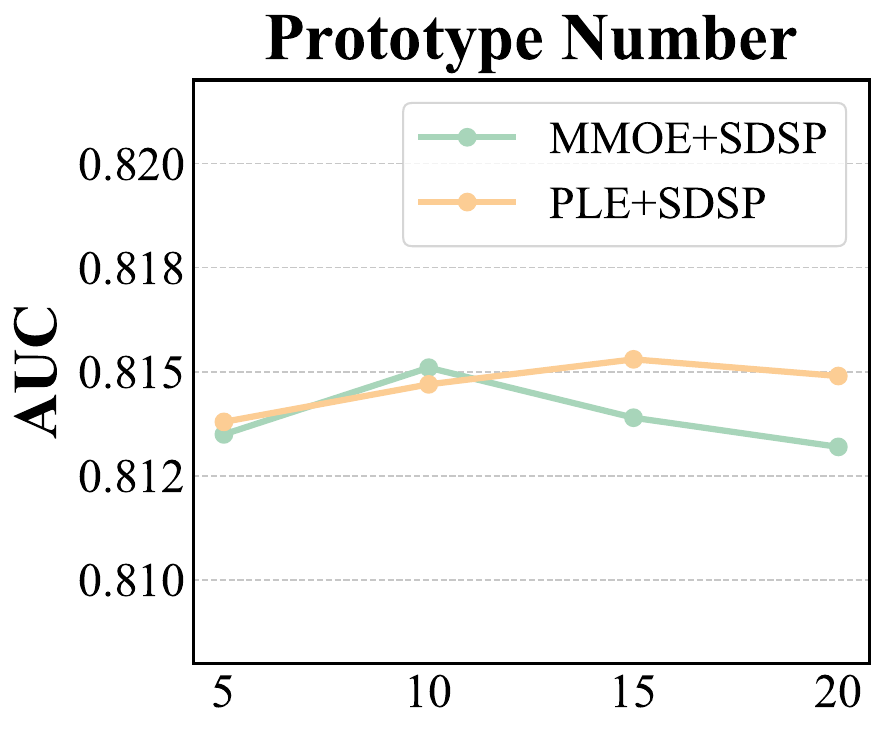}}
       \end{subfigure}
    \caption{Hyperparameter Analysis of MMOE+SDSP and PLE+SDSP on Movielens dataset.}
    \label{Hyper}
\end{figure*}

\begin{figure}[t]
\includegraphics[width=0.9\linewidth]{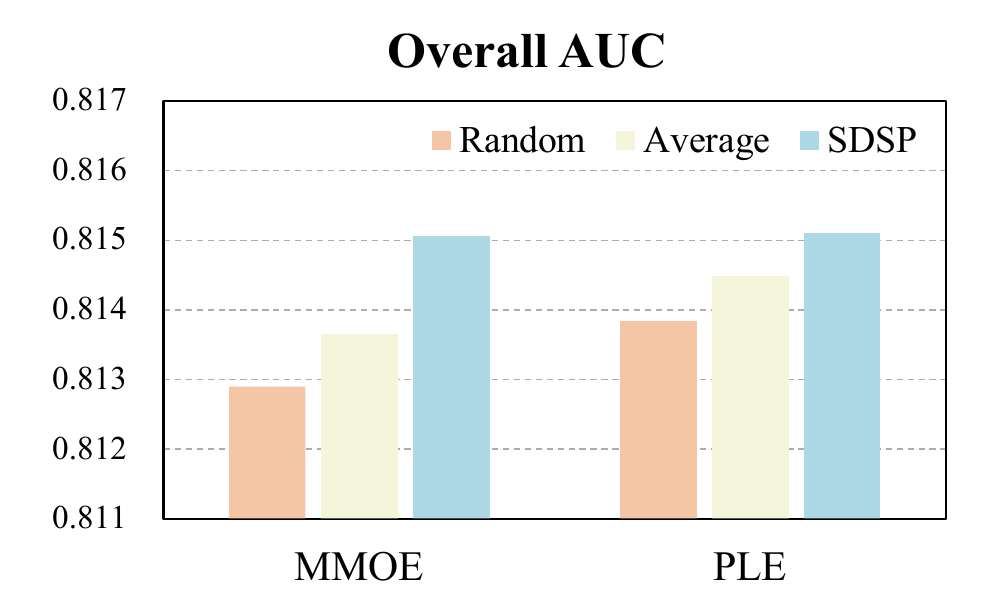}
\caption{The ablation experiments of MMOE and PLE on the Movielens dataset. ``Random'' represents completely randomly selected, and ``Average'' denotes using the average representation to measure domain distance.}
\label{ablation}
\end{figure}

\subsection{Compatibility Experiment (RQ2)}
To verify the compatibility of the proposed SDSP, we conduct experiments with three backbone models in Table \ref{Compatibility}. Based on the results, we have the following observations:
\begin{itemize}[leftmargin=*]
\item Across three datasets, integrating SDSP into backbone networks improves most AUC scores for each domain and the overall performance. For instance, in the Movielens dataset, integrating SDSP into MMOE improves the overall AUC from 0.8077 to 0.8151, which is a statistically significant increase. Similar trends are observed in other datasets and backbone network, which could attributed to the effectiveness of the SDSP strategy.
\item Our proposed SDSP improves performance differently across different baselines. Among them, MMOE shows the most significant performance improvement. We speculate that this may be due to the presence of a large number of shared structures in MMOE, leading to significant negative transfer problems in inference, which we can address nicely with domain selection.
\end{itemize}

\subsection{Hyper-Parameters Experiment (RQ3)}
To analyze the impact of various configurations on SDSP, as shown in Figure \ref{Hyper}, we explore a range of parameters, including prototype learning hyperparameter $\gamma$, expert number, batch size and prototype number. In order to ensure fairness, we fix all other parameters as default parameters, which are shown in Appendix \ref{appendix_default}.
\subsubsection{\textbf{Prototype Learning Hyperparameter $\gamma$.} }
From Figure \ref{Hyper}, both MMOE+SDSP and PLE+SDSP achieve the best performance at a learning rate of 
10$^{-4}$, with a slight decline when the learning rate is increased to 10$^{-3}$ or decreased to 10$^{-5}$. Notably, even at suboptimal learning rates, our models consistently outperform baseline methods, demonstrating the robustness of the SDSP strategy.

\subsubsection{\textbf{Expert Number.}}
To analyze the impact of expert numbers on SDSP, we report results for expert numbers 3 to 6. As shown in Figure \ref{Hyper}, MMOE and PLE exhibit slight fluctuations with changes in the expert number, and overall performance is relatively stable, indicating that our method is not sensitive to the expert number.

\subsubsection{\textbf{Batch Size.}}
 The third results in Figure \ref{Hyper} show that both MMOE+SDSP and PLE+SDSP achieve better performance at a batch size of 4096. Both MMOE+SDSP and PLE+SDSP maintain around 0.815 AUC. This may be because a larger batch size provides sufficient samples for each domain, enabling more robust learning domain distance and improved overall performance.

 \subsubsection{\textbf{Prototype Number.}}
The final experiment in Figure \ref{Hyper} evaluates the effect of prototype number on AUC performance for MMOE+SDSP and PLE+SDSP. Both methods achieve better performance around 10 prototypes, with MMOE+SDSP reaching approximately 0.815 AUC. This suggests that using 10 prototypes provides an optimal balance, effectively capturing domain features without overfitting  across excessive prototypes.

\subsection{Discussion of Model Variants (RQ4)}
To investigate the effectiveness of the proposed prototype-based metric and dynamic selection strategy, we conduct ablation experiments in Figure \ref{ablation}. In the experiment, ``Random'' represents randomly selecting similar domains for each domain, while ``Average'' denotes using the mean of the representations to measure distance for selection. The results show that both modules contribute to performance improvement, demonstrating their effectiveness. Additionally, we can find even if the ``Average''   strategy is problematic (e.g., it cannot model asymmetric distances between domains), it still outperforms ``Random'' domain selection, which can be attributed to that the proposed domain selection strategy is relatively robust, as it uses the performance as a supervisory signal to guide the selection process. This additional supervision can alleviate the negative impact of an inaccurate domain measure to some extent.

\begin{figure}[t]
\includegraphics[width=1\linewidth]{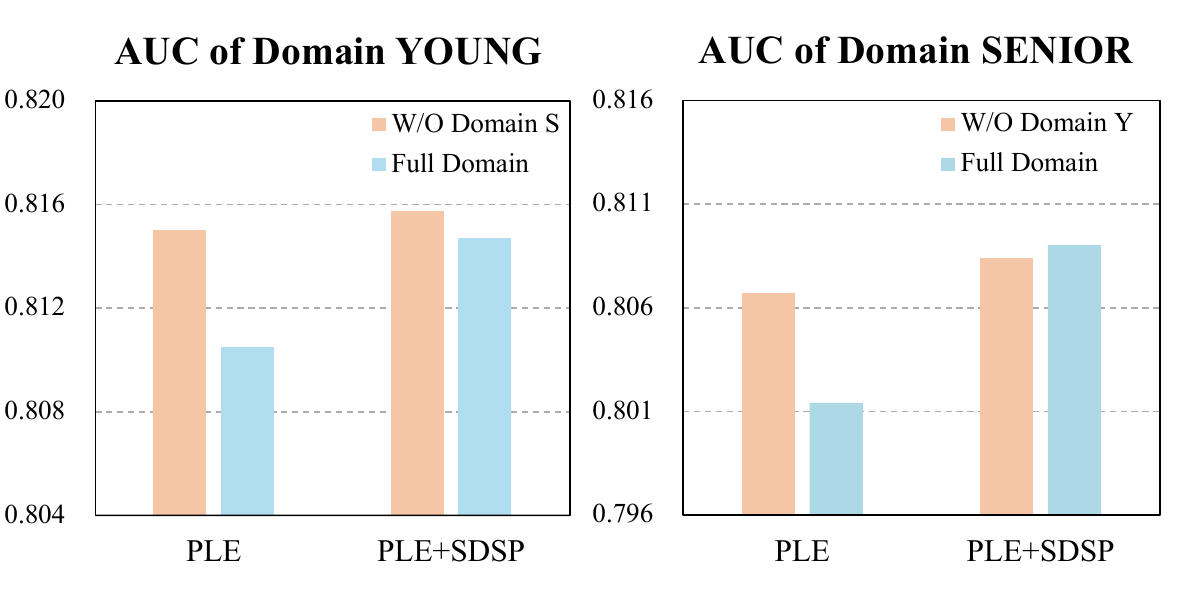}
\caption{The comparison of the negative transfer of PLE and PLE+SDSP on Movielens dataset. ``S'' and ``Y'' are abbreviations for the  SENIOR and YOUNG, respectively.}
\vspace{-3pt}
\label{negative_transfer}
\end{figure}

\subsection{Negative Transfer Research (RQ5)}
In order to investigate whether the proposed module can alleviate the negative transfer problem among domains, we conduct a negative transfer experiment of domains on the Movielens dataset, as shown in Figure \ref{negative_transfer}. To ensure the fairness of the comparison, we report on the optimal performance of different domain subsets within the same parameter search range. The original PLE model suffers from significant negative transfer issues. By incorporating the proposed SDSP, we effectively mitigate this problem. Specifically, when the domain SENIOR causes a strong negative transfer to the domain SENIOR, our method alleviates the issue. Similarly, when the domain YOUNG negatively impacts the domain SENIOR, our approach not only avoids such negative transfer but also extracts beneficial information from the senior domain to enhance the performance of the domain YOUNG.

\begin{table}[t]
\centering
\caption{Average inference time per batch with size 4096.}
\label{inference_time}
\begin{tabular}{cccc} \toprule
\multirow{2}{*}{Model} & \multicolumn{3}{c}{Average Inference Time   (ms)} \\ \cmidrule{2-4}
                       & Movielens         & Amazon        & Douban        \\ \midrule
SharedBottom           & 0.84              & 0.87          & 1.01          \\
MMOE                   & 1.20               & 0.78          & 0.97          \\
PLE                    & 1.41              & 1.12          & 1.34          \\
STAR                   & 2.78              & 1.87          & 1.98          \\
SAR-Net                & 1.93              & 1.48          & 1.74          \\
Adasparse              & 1.08              & 0.71          & 1.16          \\
Adaptdhm               & 0.96              & 0.61          & 0.82          \\
M2M                    & 2.91              & 2.71          & 2.71          \\
PPNet                  & 2.66              & 1.91          & 1.98          \\
EPNet                  & 0.61              & 0.40           & 0.52          \\ \midrule
MMOE+SDSP              & 1.59              & 1.20          & 1.40          \\
Increase               & 0.39              & 0.42          & 0.43          \\ \midrule
PLE+SDSP               & 1.96              & 1.53          & 1.72          \\
Increase               & 0.55              & 0.41          & 0.38         \\ \bottomrule 
\end{tabular}
\end{table}

\subsection{Efficiency Analysis  (RQ6)}
The time overhead of the recommendation algorithm mainly stems from the embedding size and the computational complexity. The former part we have analyzed in Section \ref{proto_measure}. Thus, we conduct the experiments to report the average inference time of the SDSP module. As shown in Table \ref{inference_time}, integrating the SDSP module into baseline models results in a minimal increase in inference time. For instance, integrating SDSP into MMOE increases inference time by 0.39ms, 0.42ms, and 0.43ms on the Movielens, Amazon, and Douban datasets, respectively. This represents the SDSP module introduces minimal inference overhead, demonstrating its efficiency. Due to space limitations, the parameter size results of the SDSP module are provided in the Appendix \ref{appendix_inference}.

\section{Related Work}


\subsection{Multi-Domain Recommendation}
Recommendation systems (RS) \cite{gu2021self,lqd1,lqd2, hyp1, hyp2, liuyue_Rec1,liuyue_rec2,wang2019ngcf-graph-rec} aims to analyze user interactions to uncover interests, becoming a key research focus in recent years. However, classical single-domain approaches are unable to process the multi-domain data, which are often encountered in real-world applications.
As a result, abundant Multi-Domain Recommendation (MDR) methods have been proposed \cite{luo2023mamdr-multi-domain-rec, chen2021user-cross-domain-rec,chang2023pepnet-multi-domain-multi-task-rec, fu2023unified-llm-multi-domain-rec, li2022gromov-cross-domain-rec,  fan2023adversarial-cross-domain-rec,gao2023autotransfer-cross-domain-rec}, leveraging shared knowledge across domains  to address challenges such as cold-start issues \cite{wang2017item,zhu2024m,jin2022multi}. These methods can be broadly categorized into Shared-Specific (SS) based methods and Dynamic Weight (DW) based methods, depending on how they model inter-domain relationships. SS-based methods\cite{tang2020ple-multi-task-rec,tong2024mdap,ning2023multi-multi-domain-graph-rec}, such as  STAR \cite{sheng2021star-multi-domain-rec} employ a shared-bottom architecture with domain-specific towers to model features. While DW-based methods \cite{yan2022apg-rec,bian2020can-rec, zhang2022m2m-multi-domain-multi-task-rec} often use scenario-sensitive features to generate  weighted parameters for the network. 

However, DW-based methods rely on manually selected features and hence are less generalizable when new scenarios are encountered. Furthermore, most SS-based approaches \cite{wang2023plate-multi-domain-rec,wang2024diff-cold-multi-domain-rec} employ a single domain-shared module, making it difficult to transfer complex multi-domain knowledge. To tackle these issues, SDSP proposes a novel domain selection module that can decouple the current single domain-shared  without additional feature engineering.

\subsection{Selection Problem}
\label{select}
Discrete selection problems are generally more challenging than continuous optimization problems. This is because discrete choices involve combinatorial complexity, where the solution space is not smooth or continuous. Thus, traditional optimization techniques like gradient-based methods cannot be directly applied, requiring specialized algorithms to explore the solution space efficiently.

In some fields, several attempts \cite{zhu2022user,zhou2022filter} have been proposed to solve different selection problems. Standley et al. \cite{standley2020tasks-multi-task} propose a group framework for choosing the suitable tasks to train together in the multi-task field.  In the multi-modal field, He et al. \cite{he2024efficient-multi-modal} proposes a greedy modality selection algorithm via submodular maximization.  In the cross-domain field,  Park et al. \cite{park2024pacer-cross-domain-rec} devise a weight factor to control the negative transfer of the multi-domain part. 

However, the greedy-based search algorithm incurs additional overhead and is not applicable in the time-sensitive field. Besides, a single gating mechanism doesn't apply to the complex multi-domain field. To address these issues, SDSP proposes a dynamic selection method to tackle the selection problem efficiently.

\section{Conclusion}
In this paper, we propose a lightweight and dynamic Similar Domain Selection Principle (SDSP) for multi-domain recommendation. Unlike relying on a single structure for domain knowledge transfer, SDSP introduces a novel prototype-based domain distance measure to model complex inter-domain relationships.  By dynamically selecting suitable domains for each target domain, SDSP alleviates NTP and enhances recommendation performance. Extensive experiments on three datasets validate the effectiveness and generalization of the proposed method, demonstrating consistent improvements across various backbone models. As the first approach to measure domain gaps and dynamically select similar domains in MDR, SDSP offers a promising solution for advancing the field.

\begin{acks}
This research was partially supported by Research Impact Fund (No.R1015-23), Collaborative Research Fund (No.C1043-24GF), Huawei (Huawei Innovation Research Program, Huawei Fellowship), Tencent (CCF-Tencent Open Fund, Tencent Rhino-Bird Focused Research Program), Alibaba (CCF-Alimama Tech Kangaroo Fund No. 2024002), Ant Group (CCF-Ant Research Fund), Kuaishou, the National Key Research and Development Program of China under Grant NO. 2024YFF0729003, the National Natural Science Foundation of China under Grant NOs. 62176014, the Fundamental Research Funds for the Central Universities, State Key Laboratory of Complex \& Critical Software Environment.
\end{acks}

\bibliographystyle{ACM-Reference-Format}
\balance
\bibliography{9-sample-base}

\appendix
\section{Theorem Proof}
\label{appendix_proof}
To proof the Theorem \ref{select_theorem}, we first need proof the following lemma:
\begin{lemma} \cite{he2024efficient-multi-modal}
\label{other_lemma}
For any domain subset $\mathbf{D} \subseteq \mathcal{D}$,  we have
$$
H(\mathbf{Y} \mid \mathbf{X}^\mathbf{D})=\inf _f \mathbb{E}_p[\ell_{\mathrm{CE}}(f(\mathbf{X}^\mathbf{D}), \mathbf{Y})]
$$
\end{lemma}

\begin{proof} 
\textbf{Lemma1} 
{
\begin{align*}
& \quad \  \mathbb{E}_p[\ell_{\mathrm{CE}}(f(\mathbf{X}^\mathbf{D}), \mathbf{Y})] \\
& =-\mathbb{E}_p[\mathbb{I}(\mathbf{Y}=0) \log (1-f(\mathbf{X}^\mathbf{D}))+\mathbb{I}(\mathbf{Y}=1) \log (f(\mathbf{X}^\mathbf{D}))] \\
& =-\mathbb{E}_D \mathbb{E}_Y[\mathbb{I}(\mathbf{Y}=0) \log (1-f(\mathbf{X}^\mathbf{D}))+\mathbb{I}(\mathbf{Y}=1) \log (f(\mathbf{X}^\mathbf{D})) \mid \mathbf{X}^\mathbf{D}] \\
& =-\mathbb{E}_D[\operatorname{Pr}(\mathbf{Y}=0 \mid \mathbf{X}^\mathbf{D}) \log (1-f(\mathbf{X}^\mathbf{D}))+\operatorname{Pr}(\mathbf{Y}=1 \mid \mathbf{X}^\mathbf{D}) \log (f(\mathbf{X}^\mathbf{D}))] \\
& = \mathbb{E}_D [\operatorname{Pr}(\mathbf{Y}=0 \mid \mathbf{X}^\mathbf{D}) \log \frac{ \operatorname{Pr}(\mathbf{Y}=0 \mid \mathbf{X}^\mathbf{D})}{ \operatorname{Pr}(\mathbf{Y}=0 \mid \mathbf{X}^\mathbf{D})} - \operatorname{Pr}(\mathbf{Y}=0 \mid \mathbf{X}^\mathbf{D}) \\
& \quad  \log (1-f(\mathbf{X}^\mathbf{D})) + \operatorname{Pr}(\mathbf{Y}=1 \mid \mathbf{X}^\mathbf{D}) \log \frac{ \operatorname{Pr}(\mathbf{Y}=1 \mid \mathbf{X}^\mathbf{D})}{\operatorname{Pr}(\mathbf{Y}=1 \mid \mathbf{X}^\mathbf{D})} \\
& \quad  - \operatorname{Pr}(\mathbf{Y}=1 \mid \mathbf{X}^\mathbf{D}) \log (f(\mathbf{X}^\mathbf{D}))]\\
& = \mathbb{E}_D [\operatorname{Pr}(\mathbf{Y}=0 \mid \mathbf{X}^\mathbf{D}) \log \frac{ \operatorname{Pr}(\mathbf{Y}=0 \mid \mathbf{X}^\mathbf{D})}{ 1 - f(\mathbf{X}^\mathbf{D})} + \operatorname{Pr}(\mathbf{Y}=1 \mid \mathbf{X}^\mathbf{D})\\
& \quad \log \frac{ \operatorname{Pr}(\mathbf{Y}=1 \mid \mathbf{X}^\mathbf{D})}{f(\mathbf{X}^\mathbf{D})} - - \operatorname{Pr}(\mathbf{Y}=0 \mid \mathbf{X}^\mathbf{D}) \log (\operatorname{Pr}(\mathbf{Y}=0 \mid \mathbf{X}^\mathbf{D})) \\
&  \quad - \operatorname{Pr}(\mathbf{Y}=1 \mid \mathbf{X}^\mathbf{D}) \log (\operatorname{Pr}(\mathbf{Y}=1 \mid \mathbf{X}^\mathbf{D}))]\\
& =\mathbb{E}_D[K L(\operatorname{Pr}(\mathbf{Y} \mid \mathbf{X}^\mathbf{D}) \| \hat{f}(\mathbf{X}^\mathbf{D},\mathbf{Y}))]+H(\mathbf{Y} \mid \mathbf{X}^\mathbf{D}) \\
& \geq H(\mathbf{Y} \mid \mathbf{X}^\mathbf{D}) .
\end{align*}}

where $\hat{f}(\mathbf{X}^\mathbf{D},\mathbf{Y})$ is the prediction distribution for $\mathbf{Y}$. Specifically,  $\hat{f}(\cdot)$ is defined as
$$
\Hat{f}(\mathbf{X}^\mathbf{D}, \mathbf{Y}) = 
\begin{cases} 
    f(\mathbf{X}^\mathbf{D}) & \text{when} \ \mathbf{Y} = 1 \\
    1 - f(\mathbf{X}^\mathbf{D}) & \text{when} \ \mathbf{Y} = 0 
\end{cases}
$$
Considering $K L(\operatorname{Pr}(\mathbf{Y} \mid \mathbf{X}^\mathbf{D}) \geq 0$, the infimum of Eq.(1) is easily to be acquire when the prediction distribution $(1 - f(\mathbf{X}^\mathbf{D}), f(\mathbf{X}^\mathbf{D}))$ equals with the true label distribution $( \operatorname{Pr}(\mathbf{Y}=0 \mid \mathbf{X}^\mathbf{D}),  \operatorname{Pr}(\mathbf{Y}=1 \mid \mathbf{X}^\mathbf{D}))$.
\end{proof}

With the Lemma \ref{other_lemma}, we can easily prove the Theorem \ref{select_theorem}.
\begin{proof}
Considering $\mathbf{Z}^i = g^j(\mathbf{X}^j)$, $ \forall j \in \mathcal{S}^i$, by the celebrated data-processing inequality, we know that
$$
I(\mathbf{Z}^i; \mathbf{Y}^i) \leq I(\mathbf{X}^j; \mathbf{Y}_i), \quad \forall j \in \mathcal{S}^i.
$$
Hence, the following chain of inequalities holds:
$$
\begin{aligned}
I(\mathbf{Z}^i; \mathbf{Y}^i) \leq \min_{j \in \mathcal{S}^i} \{I(\mathbf{X}^j; \mathbf{Y}_i)\}  \leq \max_{j \in \mathcal{D}} \{I(\mathbf{X}^j; \mathbf{Y}_i)\}  &  = I(\mathbf{X}^i; \mathbf{Y}^i) \\
&  \leq I(\mathbf{X}^1, \ldots, \mathbf{X}^D; \mathbf{Y}_i)
\end{aligned}
$$
where the last inequality follows from the fact that the joint mutual information $I(\mathbf{X}^i; \mathbf{Y}^i) \leq I(\mathbf{X}^1, \ldots, \mathbf{X}^D; \mathbf{Y}_i)$ is at least as large as any one of $I(\mathbf{X}_i^1; \mathbf{Y}_i)$ and $I(\mathbf{X}_i^2; \mathbf{Y}_i)$.
Therefore, due to the variational form of the conditional entropy, we have
\begin{align*}
    & \inf _h \mathbb{E}_p[\ell_{\mathrm{CE}}(h(\mathbf{Z}^i), \mathbf{Y}_i)]-\inf _{h^{\prime}} \mathbb{E}_p[\ell_{\mathrm{CE}}(h^{\prime}(\mathbf{X}^1,\ldots, \mathbf{X}^D), \mathbf{Y}_i)] \\
= & H(\mathbf{Y}^i \mid \mathbf{Z}^i)-H(\mathbf{Y}^i \mid \mathbf{X}^1,\ldots, \mathbf{X}^D) \\
= & I(\mathbf{X}^1,\ldots, \mathbf{X}^D ; \mathbf{Y}^i)-I(\mathbf{Z}^i ; \mathbf{Y}^i) \\
\geq & \max_{j \in \mathcal{D}} \{I(\mathbf{X}^j ; \mathbf{Y}^i)\}-\min_{j \in \mathcal{S}^i} \{ I(\mathbf{X}^j ; \mathbf{Y}^i)\} \\
= & I(\mathbf{X}^i ; \mathbf{Y}^i)-\min_{j \in \mathcal{S}^i} \{I(\mathbf{X}^j ; \mathbf{Y}^i)\}  
= \Delta_p^i 
\end{align*}
\end{proof}

\section{Experiment}

\subsection{Default Parameter} \label{appendix_default}The default parameter settings in our experiments are shown in the Table \ref{default-parameter}, where we set the default parameters for each dataset. For previous hyperparameter analyses, we have fixed the other parameters to observe the changes in the selected parameters.

\subsection{Efficiency Analysis}
\label{appendix_inference}
As shown in Table \ref{parameter}, integrating the SDSP module into baseline models results in a minimal increase in parameter size. For instance, integrating SDSP into MMOE increases parameters by only 86.05K for Movielens, from 217.64K to 303.68K. Similarly, PLE+SDSP increases by 89.77K. These results demonstrate that our method is lightweight, adding minimal parameter cost while achieving improved performance, making it suitable for resource-constrained scenarios without significant computational burden.

\begin{table}[t]
\caption{The default parameter setting in our experiments}
\label{default-parameter}
\begin{tabular}{cccc}
\toprule
Dataset                              & Movielens   & Amazon      & Douban      \\ \midrule
Learning Rate                        & 0.01        & 0.0001      & 0.001       \\
Prototype Hyper-parameter   $\gamma$ & 0.0001      & 0.0001      & 0.001       \\
Prototype Number  &10  &10 &10 \\
Decay                                & 0.9         & 0.9         & 0.9         \\
Expert Num                           & {[}1,1,1{]} & {[}1,1,1{]} & {[}1,1,1{]} \\
Batch Size                           & 4096        & 4096        & 4096    \\ \bottomrule   
\end{tabular}
\end{table}

\begin{table}[t]
\fontsize{8}{9.7}\selectfont 
\centering
\caption{Parameter Size.}
\vspace{-1em}
\label{parameter}
\begin{tabular}{cccc}
\toprule
\multirow{2}{*}{Model} & \multicolumn{3}{c}{Parameter Amount (K)} \\ \cmidrule{2-4}
                       & Movielens    & Amazon      & Douban      \\ \midrule
SharedBottom           & 227.59        & 2224.15     & 3433.38    \\
MMOE                   & 217.64        & 2218.81     & 3425.77    \\
PLE                    & 223.91        & 2220.92     & 3429.93    \\
STAR                   & 308.63        & 2277.63     & 3500.43    \\
SAR-Net                & 239.21        & 2226.94     & 3442.73    \\
Adasparse              & 230.32        & 2223.94     & 3433.68    \\
Adaptdhm               & 257.49        & 2245.99     & 3459.32    \\
M2M                    & 569.57        & 2313.07     & 3624.32    \\
PPNet                  & 429.68        & 2365.92     & 3604.85    \\
EPNet                  & 232.33        & 2219.67     & 3431.09     \\ \midrule
MMOE+SDSP              & 303.68        & 2307.64     & 3511.82    \\
Increase               & 86.05         & 88.83       & 86.05      \\ \midrule
PLE+SDSP               & 313.68        & 2308.22     & 3518.45    \\
Increase               & 89.77         & 87.30       & 88.52      \\ \bottomrule
\end{tabular}%
\end{table}

\subsection{Domain Transfer Visualization}
\label{appendix_transfer}
To verify whether the proposed metric can accurately reflect the relationships between domains, we conducted the experiment on Figure \ref{transfer}. In the left figure, the horizontal axis indicates the domain whose performance is being evaluated, while the vertical axis represents the impact on performance when a specific domain is removed during training, i.e., the value in the $i$-th row and $j$-th column shows the performance change of the $i$-th domain when the $j$-th domain is excluded from training. The right figure depicts the domain distances measured by MMOE +SDSP at the best epoch. 

To ensure fairness, we use the same set of parameters for comparison. For the transfer performance, the more performance degradation of the $i$-th domain is caused by removing the $j$-th domain, the closer the two domains should be. As seen in the right figure, the proposed metric reflects this partial order relationship well, with the size ordering of each row being correct.

\begin{figure}[t]
    \centering
    \begin{subfigure}{0.23\textwidth}  \label{transfer_mmoe}\includegraphics[width=\textwidth]{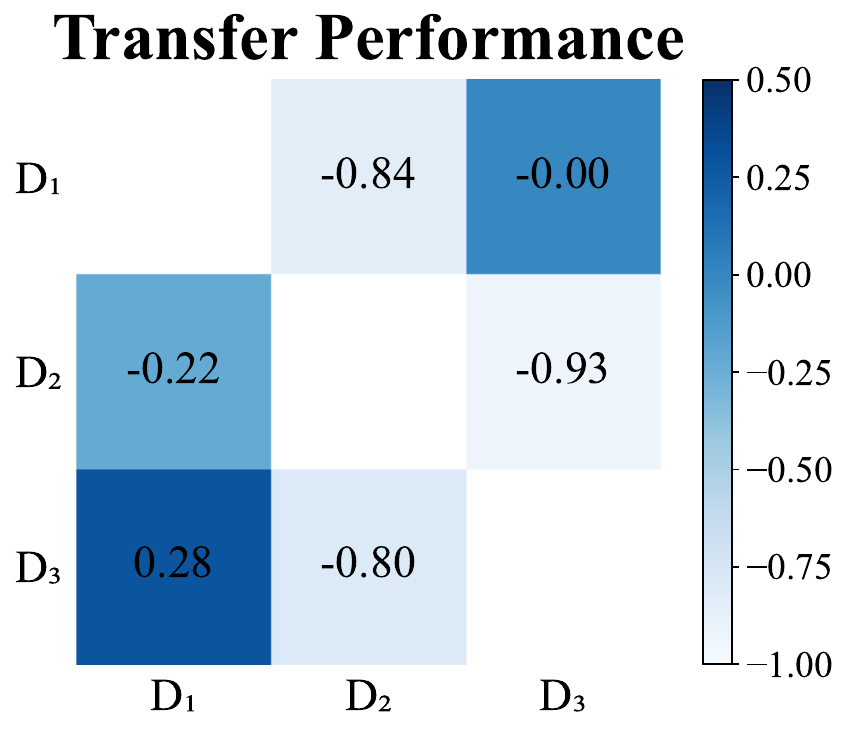}
   \end{subfigure}
    \begin{subfigure}{0.22\textwidth} 
    \label{proto_distance_mmoe}{\includegraphics[width=\textwidth]{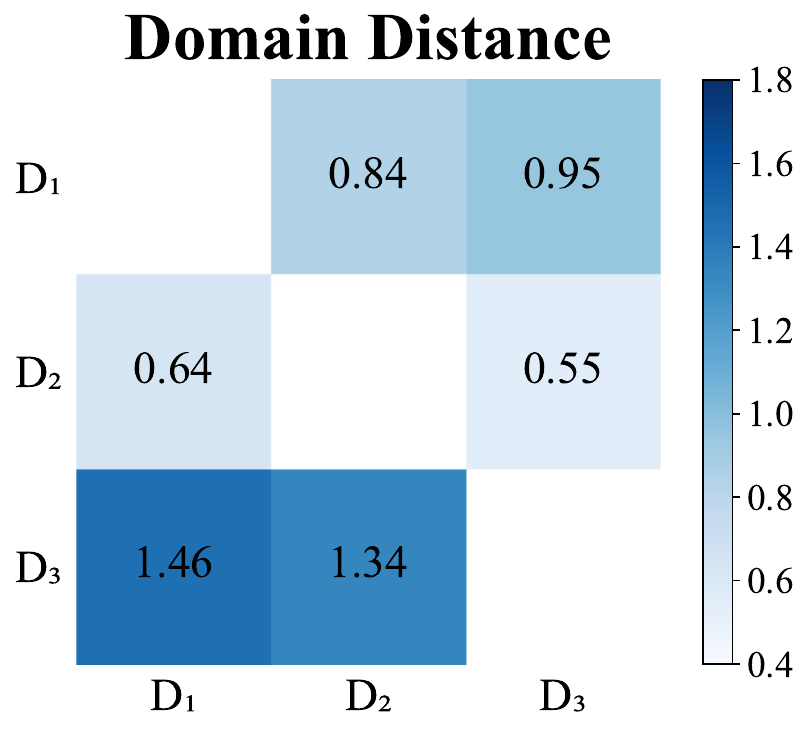}}
       \end{subfigure}
    \caption{Transfer Performance of MMOE and Domain Distance measured by MMOE+SDSP on Movielens dataset.}
    \label{transfer}
\end{figure}

\end{document}